\documentclass[pra,10pt,twocolumn,superscriptaddress,floatfix,showpacs]{revtex4-1}
\usepackage[utf8]{inputenc}
\usepackage[T1]{fontenc}     %Output what you want e.g., é, ł, a, ü
\usepackage[british]{babel}  %Do hyphenation according to british english
\usepackage[sc,osf]{mathpazo}
\usepackage[scaled=0.86]{berasans}  % URL font that go well wtih palatino
\usepackage[colorlinks=true, allcolors=blue, urlcolor=blue]{hyperref}  %Hyperlinks (pink, green, blue)
\usepackage{graphicx} % Package to insert exteral figures
\usepackage[babel]{microtype}  %Improves text justification
\usepackage{amsmath,amssymb,amsthm,bm,amsfonts,mathrsfs,bbm} %Usefull math packages

\usepackage{xspace}  %Useful to add space in macros
\usepackage{pgfplots}
\usepackage{xcolor,colortbl}
\usepackage{array}
\usepackage{bigstrut}
\usepackage{ulem}

\usepackage{caption}
\usepackage{subcaption}
\usepackage{tabularx}
\newcolumntype{C}{>{\centering\arraybackslash}X}
\usepackage{lipsum}

\newcounter{subeq}

\newcommand{\be}{\begin{equation}}
	\newcommand{\ee}{\end{equation}}
\newcommand{\bea}{\begin{eqnarray}}
	\newcommand{\eea}{\end{eqnarray}}
\newcommand{\bes}{\begin{equation*}}
	\newcommand{\ees}{\end{equation*}}
\newcommand{\beas}{\begin{eqnarray*}}
	\newcommand{\eeas}{\end{eqnarray*}}

\newtheorem{thm}{Theorem}
\newtheorem*{thm*}{Theorem}

\newtheorem{lem}[thm]{Lemma}
\newtheorem*{lem*}{Lemma}

\newtheorem*{prop*}{Proposition}

\newtheorem*{lipschitzLem*}{Lemma \ref{lipschitz}}
\newtheorem*{lipschitzCubeLem*}{Lemma \ref{lipschitzCube}}
\newtheorem*{pgmNearlyOptimalThm*}{Theorem \ref{pgmNearlyOptimal}}

\begin{document}

\title{Self-testing quantum states via nonmaximal violation in Hardy's test of nonlocality}

%########################################

\author{Ashutosh Rai}
\affiliation{School of Electrical Engineering, Korea Advanced Institute of Science and Technology, Daejeon 34141, Republic of Korea}
\affiliation{Institute of Physics, Slovak Academy of Sciences, 845 11 Bratislava, Slovakia}
%\email{ashutosh.rai@kaist.ac.kr}

\author{Matej Pivoluska}
\affiliation{Institute of Physics, Slovak Academy of Sciences, 845 11 Bratislava, Slovakia}
\affiliation{Institute of Computer Science, Masaryk University, 602 00 Brno, Czech Republic}
%\email{pivoluskamatej@gmail.com}

\author{Souradeep Sasmal}
\affiliation{Light and Matter Physics, Raman Research Institute, Bengaluru 560080, India}
%\email{souradeep.007@gmail.com}

\author{Manik Banik}
\affiliation{School of Physics, IISER Thiruvananthapuram, Vithura, Kerala 695551, India}
%\email{manik.banik@iisertvm.ac.in}

\author{Sibasish Ghosh}
\affiliation{Optics \& Quantum Information Group, The Institute of Mathematical Sciences, HBNI, C.I.T. Campus, Taramani, Chennai 600113, India}
%\email{sibasish@imsc.res.in} 

\author{Martin Plesch}
\affiliation{Institute of Physics, Slovak Academy of Sciences, 845 11 Bratislava, Slovakia}
\affiliation{Institute of Computer Science, Masaryk University, 602 00 Brno, Czech Republic}
%\email{martin.plesch@savba.sk}

%########################################

\begin{abstract}
Self-testing protocols enable certification of quantum devices without demanding full knowledge about their inner workings. A typical approach in designing such protocols is based on observing nonlocal correlations which exhibit maximum violation in a Bell test. We show that in Bell experiment known as Hardy’s test of nonlocality not only the maximally nonlocal correlation self-tests a quantum state, rather a nonmaximal nonlocal behavior can serve the same purpose. We, in fact, completely characterize all such behaviors leading to self-test of every pure two qubit entangled state except the maximally entangled ones. Apart from originating a novel self-testing protocol, our method provides a powerful tool towards characterizing the complex boundary of the set of quantum correlations.
\end{abstract}

%\pacs{03.65.Ud, 02.50.Le, 03.67.Ac}

\maketitle
\section{Introduction}
Learning quantum properties of an unknown physical system is essential for designing and testing devices based on the laws of quantum mechanics. Complete information about the physical state~\cite{QST,Gross2010} and process~\cite{Brien2004,Mohseni2008} of such a device can be obtained through tomography which requires considerable resources for implementation. On the other hand, some particular properties of a quantum system, like certification of quantum entanglement or incompatibility of measurements, can be learnt with less resources by constructing suitable witness operators~\cite{Guhne2009,Horodecki2009,Carmeli2019,Bae2020,Bavaresco2018}. The aforesaid methods still might not be optimal or available in some scenarios \cite{Reichardt2012,Vazirani2014,Miller2016}. Interestingly, however, some physical systems can be certified by employing comparatively far less resources through \textit{device-independent} tests where a device is treated simply as a black box~\cite{Scarani2012,Pironioetal2016, AcinandNavascues2017}. Then only from the input-output statistics termed as \textit{correlation} or \textit{behavior} of the box one can find the quantum state of the device. Such a certification is referred to as a self-test as it enables a user to verify the device without knowing details of its inner-workings~\cite{Mayers2004,McKague2011,YN2013,Bamps-Pironio2015,Wang-Wu-Scarani2016,Coladangeloetal,Coopmanset.al.2019,Jed2020,Ishizaka2020,Rabelo2012,DiCabello2021,Baccarietal2020,SupicBowles2020}. 

Bell inequalities naturally fit into the device-independent paradigm for certification of quantum systems since their derivations, based on the assumptions of local-realism~\cite{EPR1935}, are independent from a quantum description of the physical state for the system or measurements applied to it~\cite{Bell1964,CHSH,Hardy1993,MJWHall2011,Scarani2019,Brunneretal2014}. 
Therefore, various kind of Bell inequalities play a central role in the construction of self-testing protocols~\cite{Mayers2004,McKague2011,YN2013,Bamps-Pironio2015,Wang-Wu-Scarani2016,Coladangeloetal,Coopmanset.al.2019,Jed2020,Ishizaka2020,Rabelo2012,DiCabello2021,Baccarietal2020,SupicBowles2020}. In a Bell test only some special type of input-output statistics can self-test a quantum state; the one which can be realized (up to local isometry) by a unique quantum state and measurements~\cite{SupicBowles2020}. 
Further, within the (convex) set of all quantum correlations in a specified Bell scenario, behaviors leading to self-testing protocols must necessarily be  extremal points of the quantum set~\cite{Rabelo2012,DiCabello2021,SupicBowles2020}. 
Then the fact that any linear Bell inequality is maximized at some extremal point of the quantum set leads to a natural intuition that self-testing occurs on achieving maximal violation~\cite{SupicBowles2020,YN2013,Bamps-Pironio2015,Wang-Wu-Scarani2016,Coladangeloetal,Coopmanset.al.2019,Jed2020,Ishizaka2020,Rabelo2012,Baccarietal2020,DiCabello2021}. 

In our work, we present a different approach for self-testing quantum states. By considering a Bell experiment called Hardy's test of nonlocality~\cite{Hardy1993} we show that self-test of a quantum state is possible also with a non-maximal violation. Such self-tests can be achieved for a two-parametric set of pure qubit states, which cover the whole spectrum from almost-not-entangled to almost-fully-entangled states. This result is interesting per se, as it allows self-testing of a broad spectrum of states through a single Bell test. On top of that, it has two important corollaries. First, the presented method using the concave cover approach is very general and can also be applied to other types of Bell tests. Second, as the behaviors which lead to self-testing are extremal points of the set of all quantum correlations, our result forms a stepping stone in the way of characterizing this set~\cite{Gohetal2018,Ishizaka2018,QVoid2019} and provides an important tool for future work in this area; one such step is realized, as we find that a part of the boundary of the quantum set is determined from Hardy's nonlocal correlations.

The subsequent parts of the paper are organized as follows. In Sec. II we first introduce Hardy’s test of nonlocality and then in Sec. III we characterize all two-qubit states that can demonstrate Hardy’s nonlocality. Sec. IV contain the derivation of the main results of this paper which follows by considering Hardy's test in a black-box scenario where the dimension of quantum state of the box is unknown. Finally, in the concluding Sec. V we provide a discussion and summary of our work.

\section{Hardy's Test}
The nonlocality test proposed by Hardy~\cite{Hardy1993} relates to a Bell experiment with two space-like separated parties, Alice and Bob, who share parts of a composite physical system. Alice randomly chooses to perform one of measurements $x\in\{A_0,A_1\}$ and Bob randomly chooses to perform one of measurements $y\in\{B_0,B_1\}$ on their respective parts. 
Outcomes of all the measurements are binary, denoted $a\in\{\pm1\}$ for Alice and $b\in \{\pm 1\}$ for Bob. 
Result of the experiment termed a behavior (correlation) is recorded in a vector of probabilities: $\vec{\mathscr{P}}=\{p(a,b\vert x,y): ~x\in \{A_0,A_1\},~y\in\{B_0,B_1\},~\mbox{and}~ a,b\in \{\pm1\}\}$. A behavior is termed \textit{local}, if it can be expressed in factorized form, {\it i.e.}, $p(a,b\vert x,y)=\int_{\lambda \in \Lambda} d\lambda~p(\lambda)p(a|x,\lambda)p(b|y,\lambda)$, where $p(\lambda)$ is a probability distribution over a set of local-hidden-variables $\Lambda$, and $p(a|x,\lambda),p(b|y,\lambda)$ are local response functions of Alice and Bob respectively, which without loss of any generality can be considered as deterministic \cite{Fine1982}. 
Any behavior that cannot have a local-hidden-variable model is called \textit{nonlocal}. 
Hardy showed that if the four conditions
\begin{subequations}
\begin{align}
	p_{\mbox{\tiny Hardy}}\equiv p(+1,~+1~\vert~ A_0,~B_0)&>0, \label{eq1a}\\
	p(+1,~-1~\vert~ A_0,~B_1)&=0, \label{eq1b}\\
	p(-1,~+1~\vert~ A_1,~B_0)&=0,\label{eq1c}\\
	p(+1,~+1~\vert~ A_1,~B_1)&=0,\label{eq1d}
\end{align}
\end{subequations}
are satisfied, then the resulting behavior is necessarily nonlocal. The probability $p_{\mbox{\tiny Hardy}}$ in Eq.~(\ref{eq1a}) quantifies the amount of nonlocality and it attains the maximum value {\small $p^{\tiny \mbox{max}}_{\tiny \mbox{Hardy}}=(-11 + 5 \sqrt{5})/2$} in quantum mechanics; the maximum is achieved with projective measurements on a pure two qubit state~\cite{Rabelo2012,SeshadreesanGhosh2011}. 

\section{Two-qubit states showing Hardy's nonlocality}
In a two qubit state space {\small $\mathbb{C}^2\otimes \mathbb{C}^2$}, Hardy's nonlocal behaviors can result only from projective measurements on pure entangled states~\cite{GKar1997}. Due to local-unitary equivalence of measurements on suitably rotating the state by applying local unitary maps, without loss of generality,  let us consider projective measurements of the following form
\begin{subequations}
\begin{align}
A_0&=\vert0\rangle \langle 0\vert -\vert1\rangle \langle 1\vert,~~~~~A_1=\vert u_0\rangle\langle u_0\vert-\vert u_1\rangle\langle u_1\vert;\\
B_0&=\vert0\rangle \langle 0\vert -\vert1\rangle \langle 1\vert,~~~~~B_1~=\vert v_0\rangle\langle v_0\vert-\vert v_1\rangle \langle v_1\vert;
\end{align}\label{eq2}
\end{subequations}
where,
\begin{align*}
&\vert u_0\rangle=C_\alpha\vert 0\rangle + e^{\iota \phi}S_\alpha\vert 1\rangle,~\vert u_1\rangle=-S_\alpha\vert 0\rangle + e^{\iota \phi}C_\alpha\vert 1\rangle,\nonumber\\
&\vert v_0\rangle=C_\beta\vert 0\rangle + e^{\iota \xi}S_\beta\vert 1\rangle,~\vert v_1\rangle=-S_\beta\vert 0\rangle + e^{\iota \xi}C_\beta\vert 1\rangle, \nonumber
\end{align*}
with $C_z:=\cos(z/2),~S_z:=\sin(z/2)$ and $\alpha,\beta\in[0, \pi]$ and $\phi, \xi \in [0, 2\pi)$. Then, any pure two qubit state $\vert \psi\rangle$ satisfying constraints of Eqs.~(\ref{eq1b}-\ref{eq1d}) must be orthogonal to $\vert 0\rangle \!\otimes\!\vert v_1\!\rangle$, $\vert u_1\!\rangle \!\otimes\! \vert 0\rangle$, and $\vert u_0\!\rangle \!\otimes\! \vert v_0\!\rangle$, and therefore it must be of the form
\begin{align}
\vert \psi\rangle_{\tiny \mbox{Hardy}}\!=\!\frac{T_\alpha\vert u_0v_1\!\rangle\!+\! T_\beta\vert u_1v_0\!\rangle\!+\!\vert u_1v_1\!\rangle}{\sqrt{1+T^2_\alpha+T^2_\beta}};~~~T_z:=\tan\frac{z}{2}.\label{eq3}
\end{align}
\normalsize
The set of sixteen probabilities derived from the two qubit state in Eq. (\ref{eq3}) and measurements given by Eq.\eqref{eq2} can be expressed in an array as follows 
\begin{widetext}
{ \begin{align}
\vec{\mathscr{P}}_{\tiny \mbox{Hardy}}	
\equiv
\begingroup
\setlength{\tabcolsep}{20pt} % Default value: 6pt
\renewcommand{\arraystretch}{2} % Default value: 1
\begin{array}{c||c|c|c|c|} 
 & \!(\!+,+\!)\! & \!(\!+,-\!)\!&\!(\!-,+\!)\!&\!(\!-,-\!)\!\\ \hline\hline
\!A_0B_0\!\! &\!\dfrac{(\!1\!-\!r\!) r (\!1\!-\!s\!) s}{1\!-\!r s} \!&\!\dfrac{(\!1\!-\!r\!)^2 s}{1\!-\!r s} \!&\!(\!1\!-\!r\!)\!(\!1\!-\!s\!)\!&\!r \!\!\\\hline 
\!A_0B_1\!\! & \!(1\!\!-\!\!r) s & 0 \!&\!\dfrac{(\!1\!-\!r\!) r s^2}{1\!-\!r s} \!&\!\dfrac{1\!-\!s}{1\!-\!r s}\!\!\\\hline
\!A_1B_0\!\! &\! \dfrac{(\!1\!-\!r\!) (\!1\!-\!s\!)}{1\!-\!r s} \!&\! \dfrac{r (\!1\!-\!s\!)^2}{1\!-\!r s} \!&\! 0 \!&\! s\!\!\\\hline
\!A_1B_1\!\! &\! 0 \!&\! 1-s \!&\!\dfrac{(\!1\!-\!r\!) s}{1\!-\!r s}\!&\! \dfrac{r (\!1\!-\!s\!) s}{1\!-\!r s}\!\!\\\hline
\end{array}
\endgroup
\label{eq4}
\end{align}}
\end{widetext}
where $r:=1-S^2_\alpha S^2_\beta$ and $s:=r^{-1}C^2_\alpha$.
Note that $0\leq r,s\leq 1$ for any choice of $\alpha,\beta \in [0,\pi]$, and a behavior $\vec{\mathscr{P}}_{\tiny \mbox{Hardy}}$ is nonlocal if and only if $(r,s) \in (0,1)\times (0,1)$. The value of $\alpha$ and $\beta$ for a given value of $r$ and $s$ can be obtained from
{\small\begin{align}
\alpha=2\sin^{-1} \sqrt{1-rs},~~~~\beta=2\sin^{-1} \sqrt{(1-r)/(1-rs)}. \label{eq5}
\end{align}}
The two qubit state in Eq. (\ref{eq3}) when expressed in the standard basis, and in terms of the parameters $r$ and $s$, takes the following form
{\small	\begin{eqnarray}
\!\!\vert\!\psi(r,\!s)\!\rangle_{\tiny \mbox{Hardy}}\!\!=\!-\!\sqrt{\!\frac{(1\!-\!r) r (1\!-\!s) s}{1\!-\!r s}}~\vert00\rangle\!&-& e^{\iota \xi}\!\sqrt{\!\frac{(1\!-\!r)^2 s}{1\!-\!r s}}~\vert 01\rangle \nonumber \\
	-e^{\iota \phi}\sqrt{\!(\!1\!-\!r\!) (\!1\!-\!s\!)}~\vert 10\rangle\!&+&e^{\iota (\xi+\phi)}\sqrt{\!r}~\vert 11\rangle.\label{eq6}
\end{eqnarray}}
Hardy nonlocal states in Eq.(\ref{eq3}) [or Eq.(\ref{eq6})] covers all (up to local-unitary rotations of basis) pure two qubit entangled states except the maximally entangled ones~\cite{Goldstein1994,Jordan1994}. Moreover, in the two qubit space, due to the constraints in Eq.~(\ref{eq1a}-\ref{eq1b}), these states are uniquely determined from any arbitrarily fixed measurements~\cite{Jordan1994}. 

\section{Hardy’s test with unknown state of arbitrary dimension}
Let us now consider a black-box experiment (under i.i.d. assumption) where the quantum state and measurements of the box are unknown to Alice and Bob. The main result of this paper is that if the probabilites from the black-box are in the form given by Eq.~(\ref{eq4}), then they are a self-test of the two qubit state in Eq.~(\ref{eq6}). The black-box experiment records $p(-1,-1\vert A_0,B_0)\equiv\mathbf{r}$ and $p(-1,-1\vert A_1,B_0)\equiv\mathbf{s}$, and then verifies if the remaining probabilities are expressible in terms of $\mathbf{r}$ and $\mathbf{s}$ in the form of Eq.(\ref{eq4}); if true, then the claim is that, it is a self-test of the two qubit state $\vert \psi(\mathbf{r},\mathbf{s})\rangle_{\tiny \mbox{Hardy}}$ in Eq.~(\ref{eq6}). 

The main idea in the proof of our claim follows from a Jordan canonical form for measurements $A_0, A_1$ ($B_0,B_1$) of Alice (Bob), and application of Jensen's inequality. Suppose some unknown state $\rho$ is shared between Alice and Bob and $\Pi_{a\vert x}$ ($\Pi_{b\vert y}$) is the measurement operator associated with outcome $a$ ($b$) when Alice (Bob) measures observable $x$ ($y$). 
Then we have, $p(a,b\vert x,y)=\mbox{Tr}(\rho~ \Pi_{a\vert x} \otimes \Pi_{b\vert y})$.
Since the dimension of the state space is unrestricted, Neumark's dilation theorem allows us to reduce the analysis to projective measurements. 
Then, let the observable of Alice and Bob be some Hermitian operators with eigenvalues in $\{+1,-1\}$ as follows
{\small \begin{eqnarray*}
	x&=& (+1)~\Pi_{+\vert x} ~+~(-1)~\Pi_{-\vert x}~~~~\mbox{where}~x\in\{A_0,A_1\},\\
	y&=& (+1)~\Pi_{+\vert y} ~+~(-1)~\Pi_{-\vert y}~~~~\mbox{where}~y\in\{B_0,B_1\}.
\end{eqnarray*}}
These observables can be written in a Jordan canonical form derived in Ref.~\cite{Masanes2006} (also see~\cite{Masanes2005,Pironio2009}), which states that: for any set of four projection operators {\small $ \{\Pi_{+\vert \mathcal{M}_0},~\Pi_{-\vert \mathcal{M}_0},~\Pi_{+\vert \mathcal{M}_1},~\Pi_{-\vert \mathcal{M}_1}\}$} acting on a Hilbert space $\mathcal{H}$ and satisfying conditions {\small $\Pi_{+\vert \mathcal{M}_0} +\Pi_{-\vert \mathcal{M}_0}=I$} and {\small $\Pi_{+\vert \mathcal{M}_1} +\Pi_{-\vert \mathcal{M}_1}=I$}, there is an orthonormal basis of {\small $\mathcal{H}$} in which all the four operators are simultaneously block diagonal with maximum block size {\small $2\times 2$}. The orthonormal basis induces a direct sum decomposition {\small $\mathcal{H}=\oplus_k\mathcal{H}^k$} where dimension of each component subspace {\small $\mathcal{H}^k$} is at most two.
Then, the four projection operators can be decomposed as {\small$\Pi_{\pm\vert \mathcal{M}_{0(1)}}=\oplus_k~\Pi_{\pm\vert \mathcal{M}_{0(1)}}^k$} and each component {\small$\Pi_{\pm\vert \mathcal{M}_{0(1)}}^k$} acts on the subspace $\mathcal{H}^k$. 
The projector on subspace $\mathcal{H}^k$ can be written as {\small$\Pi^k=\Pi_{+\vert \mathcal{M}_{0}}^k+\Pi_{-\vert \mathcal{M}_{0}}^k=\Pi_{+\vert \mathcal{M}_{1}}^k+\Pi_{-\vert \mathcal{M}_{1}}^k$}. 
On applying the stated result to the observables {\small$A_0,A_1~(B_0,B_1)$} and state space {\small$\mathcal{H}_A ~\left(\mathcal{H}_B\right)$} of Alice (Bob) gives
{\small \begin{equation}
p(\!a,\!b\vert x,\!y\!)=\!\sum_{i,j} \mu_{ij}\mbox{Tr}(\rho_{ij}\Pi^i_{a\vert x}\!\! \otimes\! \Pi^j_{b\vert y})\equiv \!\sum_{i,j} \mu_{ij}p_{ij}(\!a,\!b\vert x,\!y\!), \label{eq7}
\end{equation}}
where {\small $\mu_{ij}\!=\!\mbox{Tr}(\rho\Pi^i\!\otimes\!\Pi^j)$} and satisfies {\small$\sum_{i,j}\mu_{ij}=1$} and {\small $\mu_{ij}\geq 0$}, and {\small$\rho_{ij}\!=\!(\Pi^i\!\otimes\!\Pi^j\rho\Pi^i\!\otimes\!\Pi^j)/\mu_{ij}$} is trace one positive operator on component subspace {\small $\mathcal{H}_A^i\otimes \mathcal{H}_B^j$}.

Now let us define from the probabilities in Eq.\eqref{eq4} a function {\small $\Omega(\mathbf{r},\mathbf{s})=\sum_{a,b,x,y} c_{abxy} ~p(a,b\vert x,y)+c_0$} where $c_{abxy}$ and $c_0$ are some real coefficients. On applying Eq.(\ref{eq7}) to each probability term in {\small $\Omega(\mathbf{r},\mathbf{s})$} we get
{\small\begin{align}
\Omega (\mathbf{r},\mathbf{s})= \sum_{i,j} \mu_{ij}~ \Omega (r_{ij}, ~s_{ij}), \label{eq8}
\end{align}}
where {\small $r_{ij}\!=\!p_{ij}(\!-1,\!-1\vert A_0,B_0\!)$} and {\small $s_{ij}\!=\!p_{ij}(\!-1,\!-1\vert A_1,B_0\!)$}. From Eq.(\ref{eq7}) one can find that $\mathbf{r}=\sum_{i,j} \mu_{ij}r_{ij}$ and $\mathbf{s}=\sum_{i,j} \mu_{ij}s_{ij}$. Furthermore, when the black-box statistics satisfies the zero constraints of Hardy's test, then the same is true in every {\small $\mathcal{H}_A^i\otimes \mathcal{H}_B^j$} subspace. Now we like to state and prove the following Lemma.
\begin{lem}
 Let {\small$\mathcal{E}(\mathbf{r},\!\mathbf{s})\!: \!(0,\!1)\!\times\!(0,\!1)\!\rightarrow \!\mathbb{R}$} be a concave cover of {\small$\Omega (\mathbf{r},\mathbf{s})$}, and suppose $\mathcal{R}$ is the set of points from the domain for which {\small$\mathcal{E}(\mathbf{r},\mathbf{s})=\Omega (\mathbf{r},\mathbf{s})$}. Then for all {\small$(\mathbf{r},\mathbf{s})\in \mathcal{R}$}, and for all {\small$(i,j)$},  {\small$r_{ij}=\mathbf{r}$} and {\small$s_{ij}=\mathbf{s}$}, provided {\small$\Omega(\mathbf{r},\mathbf{s})$} is a strictly concave function of {\small$\mathbf{r}$} and {\small$\mathbf{s}$} over the region {\small$\mathcal{R}$}. \label{Lemma1}
\end{lem}
\begin{proof}
A \textit{concave cover} {\small$\mathcal{E}(\mathbf{r},\mathbf{s})$} for function {\small$\Omega(\mathbf{r},\mathbf{s})$} over its domain {\small$(0,1)\times(0,1)$} is defined as the lowest-valued concave function that overestimates or equals $\Omega(\mathbf{r},\mathbf{s})$ in its domain.
Now consider Jensen's inequality which states that for any concave real function {\small$f(x):\mathbb{R}^n\rightarrow \mathbb{R},~f(\sum_{k=1}^m p_k x_k) \geq \sum_{k=1}^m p_k f(x_k)$} where {\small$x_k\in\mathbb{R}^n$} and {\small$p_k\geq 0$} for all {\small$k$}, and {\small$\sum_kp_k=1$}. When {\small$f(x)$} is nonlinear, equality holds if and only if {\small$x_1=x_2=...=x_m$}. On applying Jensen's inequality to the function {\small $\mathcal{E}(\mathbf{r},\mathbf{s}):\mathbb{R}^2\rightarrow \mathbb{R}$} we get
{\small\begin{eqnarray}
\mathcal{E}(\mathbf{r},\mathbf{s})=\mathcal{E}(\sum_{i,j} \mu_{ij}~( r_{ij}, ~s_{ij}))
\geq \sum_{ij} \mu_{ij}~ \mathcal{E}(r_{ij},~ s_{ij}). \label{eq9}
\end{eqnarray}}
On the other hand, when {\small $(\mathbf{r},\mathbf{s}) \in \mathcal{R}$} we have {\small$\mathcal{E}(\mathbf{r},\mathbf{s}) =\Omega (\mathbf{r},\mathbf{s})$}, and since {\small$ \Omega(r_{ij},s_{ij})\leq\mathcal{E}(r_{ij},s_{ij})$}, using Eq.\eqref{eq10}, we get
{\small\begin{align}
\mathcal{E}(\mathbf{r},\mathbf{s}) =\sum_{i,j} \mu_{ij}~\Omega(r_{ij}, s_{ij}) 
 \leq \sum_{i,j} \mu_{ij}~\mathcal{E}(r_{ij}, s_{ij}).\label{eq10}
\end{align}}
Then, Eq.~(\ref{eq9}) and Eq.~(\ref{eq10}) imply that if $(\mathbf{r},\mathbf{s})\in \mathcal{R}$
{\small\begin{align}
\mathcal{E}(\sum_{i,j} \mu_{ij}( r_{ij}, ~s_{ij})) = \sum_{ij} \mu_{ij}~ \mathcal{E}(r_{ij}, s_{ij}). \label{eq11}
\end{align}}
Thus in Jensen's inequality~(\ref{eq9}), equality holds in the region $\mathcal{R}$. Therefore, if $\Omega(\mathbf{r},\mathbf{s})$ is a strictly concave function of $\mathbf{r}$ and $\mathbf{s}$ over the region $\mathcal{R}$, in every {\small$\mathcal{H}_A^i\otimes \mathcal{H}_B^j$} subspace value of the parameters  $(r_{ij},s_{ij})$ must be same, i.e., $r_{ij}=\mathbf{r}$ and $s_{ij}=\mathbf{s}$ for all $i,j$. 
\end{proof}
\begin{figure}[t!]
	\begin{center}
		\includegraphics[angle=0, width=0.3\textwidth]{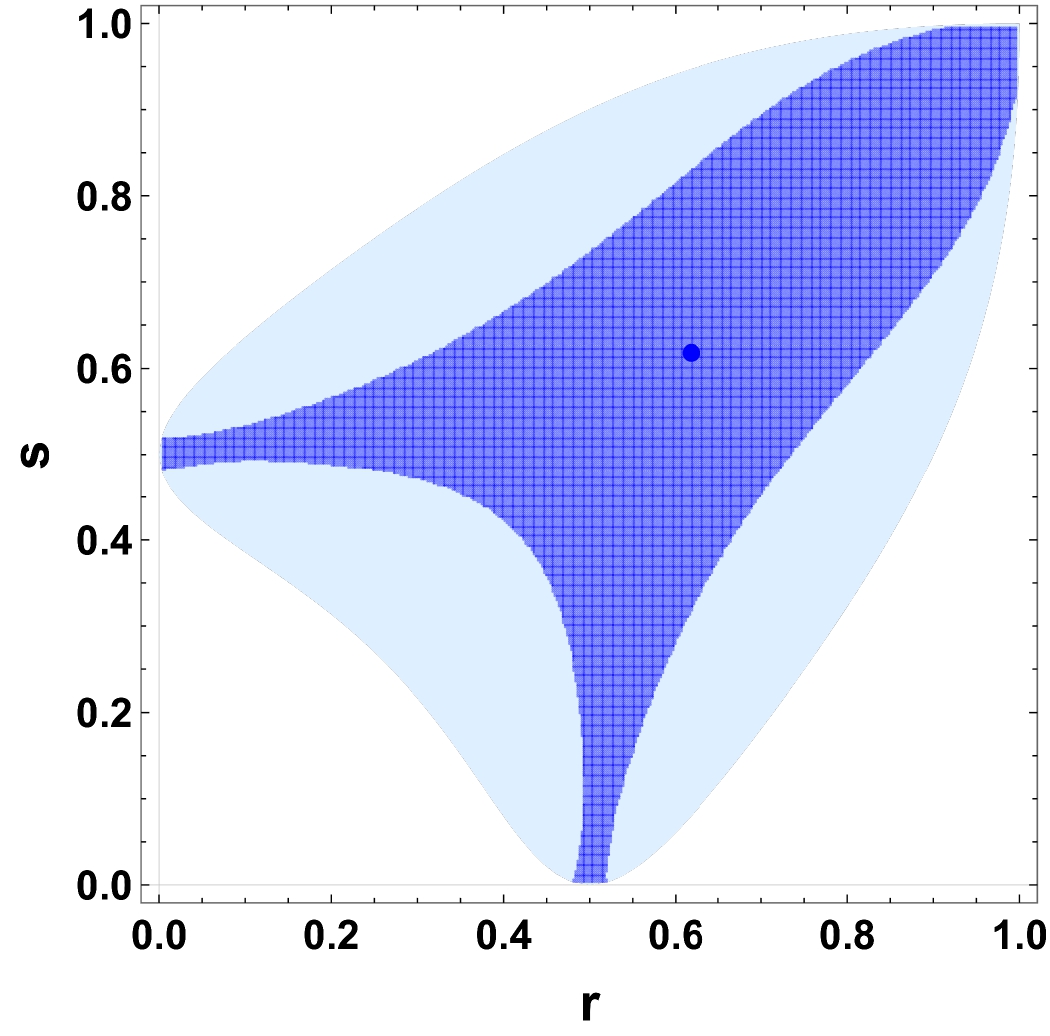}
		\caption{ In the light (blue) shaded region, {\small$\Omega^{\ast}(\mathbf{r},\mathbf{s})$} is concave. {\small$\mathcal{R}^{\ast}(\mathbf{r},\mathbf{s})$} is the dark (blue) shaded region,  where the concave cover {\small$\mathcal{E}^{\ast}(\mathbf{r},\mathbf{s})=\Omega^{\ast}(\mathbf{r},\mathbf{s})$. } The dark (blue) dot shows point $(\mathbf{r},\mathbf{s})=(\frac{\sqrt{5}-1}{2} ,\frac{\sqrt{5}-1}{2})$ where $\Omega^{\ast}(\mathbf{r},\mathbf{s})$ is maximum.}  \label{Fig1}
	\end{center}	
\end{figure}
Let us now give an example where Lemma \ref{Lemma1} is applicable. Consider a function defined by the probability which quantifies success in Hardy's test, i.e.,
{\small\begin{align}
\Omega^{\ast}(\mathbf{r},\mathbf{s})\equiv p(+1,+1\vert A_0,B_0)=\frac{\mathbf{r}(1\!-\!\mathbf{r})\mathbf{s}(1\!-\!\mathbf{s})}{1\!-\!\mathbf{r} \mathbf{s}}.
\label{eq12}
\end{align}}
The function {\small$\Omega^{\ast}(\mathbf{r},\mathbf{s})$} is a nonlinear function and it is concave in a part of its domain. Further, for the considered function, there exist a concave cover {\small $\mathcal{E}^{\ast}(\mathbf{r},\mathbf{s})$}, and a region {\small $\mathcal{R}^{\ast}\subset (0,1)\times(0,1)$} where {\small$\mathcal{E}^{\ast}(\mathbf{r},\mathbf{s})=\Omega^{\ast}(\mathbf{r},\mathbf{s})$}. 
We constructed the cover function numerically and then find the region $\mathcal{R}^{\ast}$ as shown in Fig. \ref{Fig1} (see Appendix~(A) for the details of the computational method). 
From the considered example and application of Lemma~\ref{Lemma1}, we conclude that if probability distribution of the black-box is such that $(\mathbf{r},\mathbf{s})\in \mathcal{R}^{\ast}$ then in every subspace  {\small$\mathcal{H}_A^i\otimes \mathcal{H}_B^j$, $r_{ij}=\mathbf{r}$ and $s_{ij}=\mathbf{s}$}. We now ask whether such a property can hold for all $(\mathbf{r},\mathbf{s})\in (0,1)\times (0,1)$. In other words, is it possible to vary over the possible choice of functions $\Omega(\mathbf{r},\mathbf{s})$ such that union of all the resulting regions $\mathcal{R}$ covers the full parameter space {\small$(0,1)\times (0,1)$}? Such a property will free the black box parameters $\mathbf{r}$ and $\mathbf{s}$ from any restrictions and lead to an interesting extension of Lemma~\ref{Lemma1}. We find answer to the question in affirmative through the following lemma.
\begin{lem}
 If the black-box statistics in the Bell-experiment is of the form $\vec{\mathscr{P}}_{\tiny \mbox{Hardy}}(\mathbf{r},\mathbf{s})$ as in Eq.~(\ref{eq4}), then for all $(\mathbf{r},\mathbf{s})\in(0,1)\times (0,1)$ and in all $\mathcal{H}_A^i\otimes \mathcal{H}_B^j$ subspace, $r_{ij}=\mathbf{r}$ and $s_{ij}=\mathbf{s}$. \label{Lemma2}
 \end{lem}
\begin{proof}
First we recall that Lemma \ref{Lemma1} is applicable to any function {\small $\Omega(\mathbf{r},\mathbf{s})$} defined as some linear combination of all the probabilities in Eq~(\ref{eq4}) plus a constant term. We find that, a proof of Lemma~\ref{Lemma2} follows on considering simply a single parameter family of functions
{\small \begin{align}
\Omega_{\nu}(\mathbf{r},\!\mathbf{s}) \!&=\!\Omega^{\ast}(\mathbf{r},\!\mathbf{s}) \!+\nu~\!p(\!+\vert A_0\!) \!+(1\!-\!\nu)~\!p(\!-\vert A_0\!)\!-\!1/2,\nonumber\\
	\!&=\!\Omega^{\ast}(\mathbf{r},\!\mathbf{s})\!+\nu~\!(\mathbf{s}\!-\!\mathbf{r}\mathbf{s}) \!+(1\!-\!\nu)~\!(1\!-\!\mathbf{s}\!+\!\mathbf{r}\mathbf{s})\!-\!1/2, \label{eq13}
\end{align}}
where {\small $0\leq \nu \leq 1$}, {\small $p(\!+\vert \!A_0\!)\!=\!p(\!+,\!+\vert A_0,\!B_0\!)+p(\!+,\!-\vert A_0,\!B_0\!)$, and $p(\!-\vert A_0\!)=1-p(\!+\vert A_0\!)$}. Note that at {\small$\nu =1/2$}, {\small$\Omega_{1/2}(\mathbf{r},\mathbf{s})\equiv \Omega^{\ast}(\mathbf{r},\mathbf{s})$} and {\small $\mathcal{R}_{1/2}\equiv \mathcal{R}^{\ast}$}. 
For the family of functions {\small$\Omega_{\nu}(\mathbf{r},\mathbf{s})$} on considering {\small$\nu \in \{d/N:~d\in\{1,2,...,N-1\} \}$}, we find that for a sufficiently large value of {\small$N$}, {\small$\cup_{\nu} \mathcal{R}_{\nu}= (0,1)\times (0,1)$}. To observe the proof, first let us see  Fig.~\ref{Fig2} where the region {\small$\cup_{\nu} \mathcal{R}_{\nu}$} is shown when for {\small$N=10$}, i.e., when the parameter $\nu$ takes values from the set $\{0.1,~0.2,...,0.9\}$. Note that the region where {\small$r_{ij}=\mathbf{r}$} and {\small$s_{ij}=\mathbf{s}$}, in all {\small$\mathcal{H}_A^i\otimes \mathcal{H}_B^j$} subspace, is expanded. A more detailed exposition of the proof which follows on increasing the value of {\small$N$} is provided in the Appendix~(B).
\end{proof}
Finally, in the following we prove the main result of the paper by using Lemma~\ref{Lemma2}.
\begin{figure}[t!]
	\begin{center}
		\includegraphics[angle=0, width=0.3\textwidth]{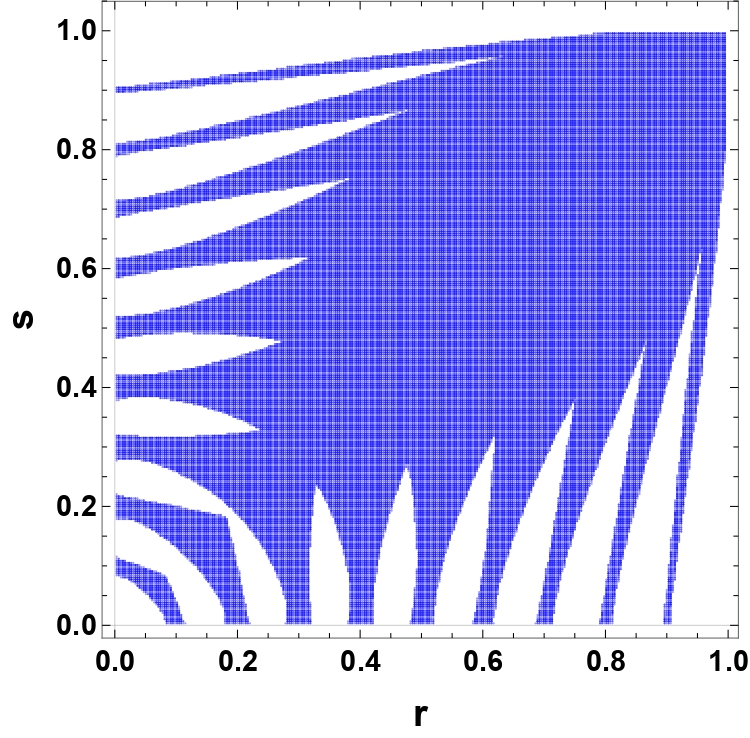}
		\caption{ The dark (blue) shaded region {\small$\cup_{\nu} \mathcal{R}_{\nu}$} is shown when the parameter {\small$\nu$} takes values from the set {\small$\{0.1,~0.2,...,~0.9\}$}. If {\small$(\mathbf{r},\mathbf{s})$} belongs to the shaded region, then {\small$(r_{ij},s_{ij})=(\mathbf{r},\mathbf{s})$} in every {\small$\mathcal{H}_A^i\otimes \mathcal{H}_B^j$} subspace. When {\small$\cup_{\nu} \mathcal{R}_{\nu}$} is derived from {\small$\{\nu=k/N:~k\in\{1,2,...,N-1\}\}$}, as {\small$N$} becomes sufficiently large the shaded region covers all {\small$(\mathbf{r},\mathbf{s}) \in (0,1)\times (0,1)$}.}  \label{Fig2}
	\end{center}	
\end{figure}
\begin{thm*}
In a black-box Bell experiment, if a behavior $\vec{\mathscr{P}}_{\tiny \mbox{Hardy}}(\mathbf{r},\mathbf{s})$ of the form given in Eq.~(\ref{eq4}) is observed, then the state of unknown system $\rho^{AB}$ is equivalent up to local isometries to $ \zeta^{AB}\otimes\vert\psi(\mathbf{r},\mathbf{s})\rangle_{\tiny \mbox{Hardy}}^{A'B'}\langle\psi(\mathbf{r},\mathbf{s})\vert$, where $\vert\psi(\mathbf{r},\mathbf{s})\rangle^{A'B'}_{\tiny \mbox{Hardy}}$ is the pure two qubit Hardy state given by the Eq.~(\ref{eq6}) and $\zeta^{AB}$ is some arbitrary bipartite state.
\end{thm*}
\begin{proof}
Consider a purification $\vert\chi\rangle^{ABP}$ of the unknown state $\rho^{AB}$, in brief we denote such a purification simply as $\vert\chi\rangle^{AB}$ since the desired local isometry is not required to act on the purification space~\cite{SupicBowles2020}. Also observables $A_0,A_1,B_0,B_1$ can be considered in a $2\times2$ block diagonal form~\cite{SupicBowles2020}. Then, there is a basis in which the observables are in the following block diagonal form
{\small\begin{align}
\Pi^i_{+\vert A_0}&=\vert2i\rangle\langle2i\vert,~~~~~~\Pi^i_{-\vert A_0}=\vert2i+1\rangle\langle2i+1\vert;\nonumber\\
\Pi^i_{+\vert A_1}&=\vert u_{2i}\rangle\langle u_{2i}\vert,~~~
\Pi^i_{-\vert A_1}=\vert u_{2i+1}\rangle\langle u_{2i+1}\vert;\nonumber\\
\Pi^j_{+\vert B_0}&=\vert2j\rangle\langle2j\vert,~~~~~~ \Pi^j_{-\vert B_0}=\vert2j+1\rangle\langle2j+1\vert;\nonumber\\
\Pi^j_{+\vert B_1}&=\vert v_{2j}\rangle\langle v_{2j}\vert,~~~ \Pi^j_{-\vert B_1}=\vert v_{2j+1}\rangle\langle v_{2j+1}\vert; 
\label{eq14}
\end{align}}
where,
{\small \begin{align}
\vert u_{2i}\!\rangle\!&=\!C_{\alpha_i}\vert 2i\rangle \!+ \!e^{\mathtt{i} \phi}S_{\alpha_i}\vert 2i\!+\!1\rangle,~
\vert u_{2i+1}\!\rangle\!=\!-S_{\alpha_i}\vert 2i\rangle \!+\! e^{\mathtt{i} \phi}C_{\alpha_i}\vert 2i\!+\!1\rangle; \nonumber\\
\vert v_{2j}\!\rangle\!&=\!C_{\beta_j}\vert 2j\rangle \!+\! e^{\mathtt{i} \xi}S_{\beta_j}\vert 2j\!+\!1\rangle,~
\vert v_{2j+1}\!\rangle\!=\!-S_{\beta_j}\vert 2j\rangle \!+\! e^{\mathtt{i} \xi}C_{\beta_j}\vert 2j\!+\!1\rangle, \nonumber
\end{align}}
with $i,j\in \{0,1,2,...\}$, $\alpha_i,\beta_j\in[0, \pi]$ and $\phi, \xi \in [0, 2\pi)$. 
Note that, without loss of generality, the phases $\xi$ and $\phi$ can be considered independent of indices $i$ and $j$ in every {\small $\mathcal{H}_A^i\otimes \mathcal{H}_B^j$} subspace since this can be achieved by choosing a suitable measurement basis and rotating the state by local unitaries such that the probabilities $p_{ij}(a,b\vert x,y)$ remain invariant. Also note that in the subspace {\small$\mathcal{H}_A^i\otimes \mathcal{H}_B^j$ parameters $r_{ij}$ and $s_{ij}$} are given by {\small$r_{ij}=1-S^2_{\alpha_i}S^2_{\beta_j},~~s_{ij}=r^{-1}_{ij}C^2_{\alpha_i}$}. 
Now suppose $\vec{\mathscr{P}}_{\tiny \mbox{Hardy}}(\mathbf{r},\mathbf{s})$ is observed in the black-box experiment, then Lemma~\ref{Lemma2} implies that in each subspace {\small$\mathcal{H}_A^i\otimes \mathcal{H}_B^j $}, $r_{ij}=\mathbf{r}$ and $s_{ij}=\mathbf{s}$ and, therefore, $\vec{\mathscr{P}}_{\tiny \mbox{Hardy}}(r_{ij},s_{ij})=\vec{\mathscr{P}}_{\tiny \mbox{Hardy}}(\mathbf{r},\mathbf{s})$. This can be true if and only if in {\small$\mathcal{H}_A^i\otimes \mathcal{H}_B^j$} subspace, $\rho_{ij}=\vert\psi_{\tiny \mbox{Hardy}}\rangle_{ij}\langle \psi_{\tiny \mbox{Hardy}}\vert$, where
{\small \begin{align}
&\vert\!\psi_{\tiny \mbox{Hardy}}\!\rangle_{ij}\!=\!-\sqrt{\!\frac{(1\!-\!\mathbf{r}) \mathbf{r} (1\!-\!\mathbf{s}) \mathbf{s}}{1\!-\!\mathbf{r} \mathbf{s}}}~\vert2i,2j\rangle\!-\!e^{\iota\xi}\!\sqrt{\!\frac{(1\!-\!\mathbf{r})^2 \mathbf{s}}{1\!-\!\mathbf{r} \mathbf{s}}}~\vert 2i,2j\!+\!1\rangle \nonumber\\
&-e^{\iota\phi}\!\sqrt{\!(1\!-\!\mathbf{r}) (1\!-\!\mathbf{s})}~\vert 2i\!+\!1,2j\rangle\!+\!e^{\iota(\xi+\phi)}\!\sqrt{\!\mathbf{r}}~\vert 2i\!+\!1,2j\!+\!1\rangle, 
\label{eq15}
\end{align}}
Hence, the unknown state {\small$\vert\chi \rangle^{AB}$} can only be a direct sum of the form {\small$\vert\chi \rangle^{AB}=\bigoplus_{i,j}\sqrt{\mu_{ij}}~\vert\psi_{\tiny \mbox{Hardy}}\rangle_{ij}$}. Finally, we can give local isometries $\Phi^A$ and $\Phi^B$ such that
{\small\begin{equation*}
	(\Phi^A\otimes\Phi^B)\vert\chi\rangle^{AB}\vert 00\rangle^{A'B'}=\zeta^{AB}\otimes\vert\psi(\mathbf{r},\mathbf{s})\rangle_{\tiny \mbox{Hardy}}^{A'B'}\langle\psi(\mathbf{r},\mathbf{s})\vert,
\end{equation*}}
where components of the $\vert 00\rangle^{A'B'}$ are local ancilla qubits appended to the unknown state $\vert\chi\rangle^{AB}$, and after application of the local isometry $\Phi^A\otimes\Phi^B$ we want to get the target state $\vert\psi(\mathbf{r},\mathbf{s})\rangle_{\tiny \mbox{Hardy}}^{A'B'}$ along with some arbitrary bipartite state $ \zeta^{AB}$. 
The isometry map $\Phi^A=\Phi^B=\Phi$ with $\Phi~\vert 2k,0\rangle^{XX'} \mapsto \vert 2k,0\rangle^{XX'}$ and $ \Phi~\vert 2k+1,0\rangle^{XX'} \mapsto \vert 2k,1\rangle^{XX'}$, where $XX'\in\{AA',BB'\}$, has the desired property. 
This concludes our proof.
\end{proof}
The Theorem implies that any behavior in the form of Eq.~(\ref{eq4}) are extremal points of the set of quantum behaviors and hence it determines a part of the quantum boundary. Further, one can check that all the behaviors that are Hardy nonlocal but not in the form given by Eq.~(\ref{eq4}) are in the interior of the quantum set and they must arise on measuring mixed entangled states in higher than two qubit dimensions. 

\section{Conclusion}
Hardy's test of nonlocality, viewed differently, is a form of tailored Bell inequality~\cite{Acin2012, Alexia2017}. However, in distinction to, for instance, \textit{tilted}-Clauser-Horne-Shimony-Holt inequality~\cite{Acin2012}, Hardy's test places certain constraints on some outcome probabilities. Importance of tailored Bell inequalities is starkly revealed through various protocols for self-testing, randomness generation, quantum key distribution, etc., wherever the knowledge of the geometry of the quantum set of correlations plays a crucial role~\cite{Gohetal2018}. In this sense tailoring by constraining outcome probabilities can possibly play an important role, as we have demonstrated in this work through the Hardy's test. Here we note that for self-testing two qubit states and measurements on them, tilted-CHSH inequality is more powerful than Hardy's test, however, Hardy's correlations can self-test almost all two qubit states from the whole spectrum of violations, ranging from arbitrarily small to maximal violation in the Hardy's test of nonlocality. The fact that our approach can have broader applications can be seen by noticing the three key features in the derivation of our results: (i) due to the constraints on probabilities a Hardy nonlocal state, with local state space $\mathbb{C}^2$, is pure and unique for an arbitrarily fixed measurement~\cite{GKar1997,Jordan1994}, (ii) existence of a simple Jordan canonical form for two dichotomic observables for each party~\cite{Masanes2006,Masanes2005,Pironio2009}, and (iii) possibility of finding a concave cover to an arbitrary linear combination of outcome probabilities and application of Jensen's inequality. Thus, there can be tailored Bell tests other than Hardy's where the three properties may hold. In general, these three features may be found in any Bell scenario with $n$-parties, two measurements per party, and two outcomes to every measurement.  

To summarize, in this work, we have presented a method for self-testing quantum states by using the concave cover of a linear combination of observed outcome probabilities. We showcased the applicability of this approach on Hardy's test of nonlocality, leading to a two-parametric set of characterizable states with a broad spectrum of entanglement properties. This technique allows a full (up to local isometries) specification of the source state, even without maximal violation in Hardy's test. Our results show many potential development paths as the method introduced in this work can be possibly applied to different kind of Bell tests, allowing better specification of source states in a device-independent fashion. At the same time, it provides a very powerful tool for characterizing the boundary of the quantum set of correlations.

\begin{acknowledgments}
AR is supported by the National Research Foundation of Korea (NRF-2021R1A2C2006309), an Institute of Information and Communications Technology Promotion (IITP) grant funded by the Korean government (MSIP) (Grant No. 2019-0-00831) and the University IT Research Center (ITRC) Program (IITP-2021-2018-0-01402). AR, M. Pivoluska, and M. Plesch acknowledge funding and support from VEGA Project~No.~2/0136/19. M. Pivoluska, and M. Plesch additionally acknowledge GAMU project MUNI/G/1596/2019.  M.B. acknowledges funding from the
National Mission in Interdisciplinary Cyber-Physical systems from the Department of Science and Technology through
the I-HUB Quantum Technology Foundation (Grant No. I-HUB/PDF/2021-22/008), support through the research grant of INSPIRE Faculty fellowship from the Department of Science and Technology, Government of India, and the start-up research grant from SERB, Department of Science and Technology (Grant No. SRG/2021/000267).

AR thankfully acknowledge fruitful discussions and feedback from  Prof. M.J.W. Hall, Prof. Joonwoo Bae, Prof. Guruprasad Kar, and Dr. Ramij Rahaman at different stages of this work.

\end{acknowledgments}

%\appendix 

\begin{widetext}

\section*{Appendix~(A): Computation of Concave cover for functions $\boldmath{\Omega(\mathbf{r},\mathbf{s})}$} \label{apnA}
Let us describe our method for computing the concave covers to the class of functions $\Omega(\mathbf{r},\mathbf{s})$ introduced in  Lemma-1 of the main text. Consider the first concrete example to which application of the Lemma-1 gives interesting insight,
\begin{align}
	\Omega^{\ast}(\mathbf{r},\mathbf{s})=\frac{\mathbf{r}(1\!-\!\mathbf{r})\mathbf{s}(1\!-\!\mathbf{s})}{1\!-\!\mathbf{r} \mathbf{s}}.
\end{align}
The three dimensional plot of $	\Omega^{\ast}(\mathbf{r},\mathbf{s})$ is shown in Fig.(\ref{SupFig1a}) which indicates that the function can be concave in some parts of its domain. On computing the two eigenvalues of Hessian matrix of $\Omega^{\ast}(\mathbf{r},\mathbf{s})$ we find that there a region of domain where both the eigenvalues are negative, therefore, the function is concave in this region; this region is shown in Fig.(\ref{SupFig1b}).

\begin{figure}[!ht]
	\begin{tabularx}{\linewidth}{CC}
		\begin{subfigure}[b]{0.4\textwidth}
			\includegraphics[width=\linewidth]{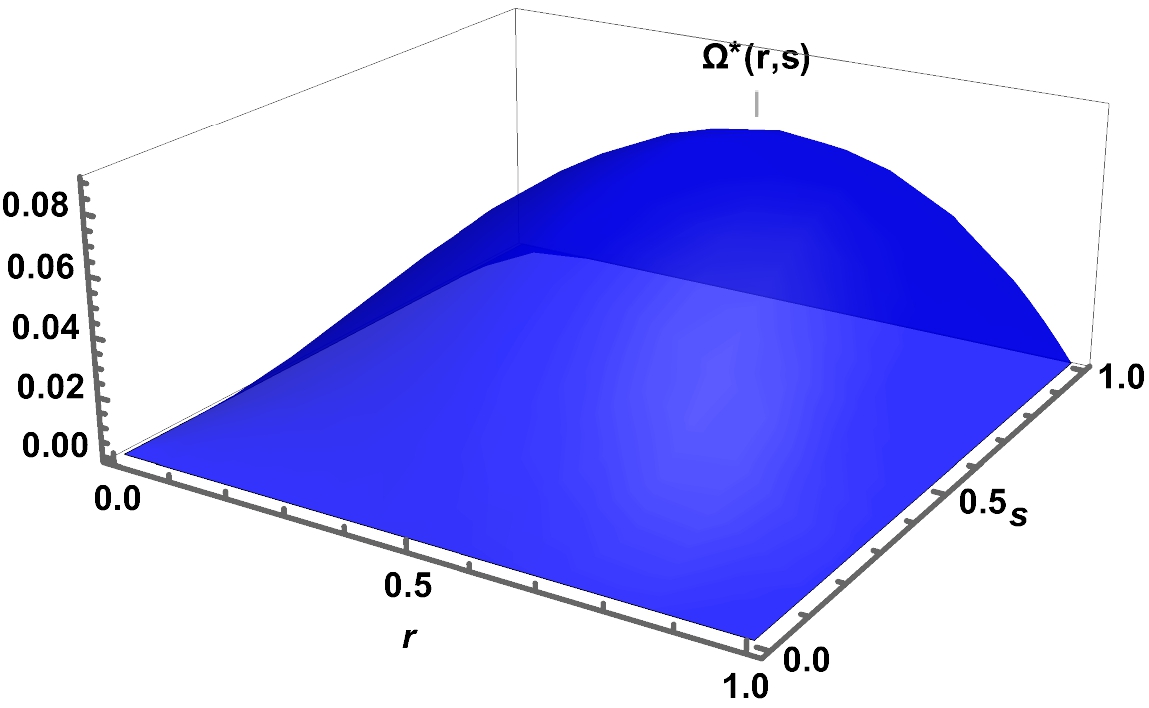}
			\caption{A $3$-dimensional  plot of $\Omega^{\ast}(\mathbf{r},\mathbf{s})$}
			\label{SupFig1a}
		\end{subfigure}
		&	
		\begin{subfigure}[b]{0.3\textwidth}
			\includegraphics[width=\linewidth]{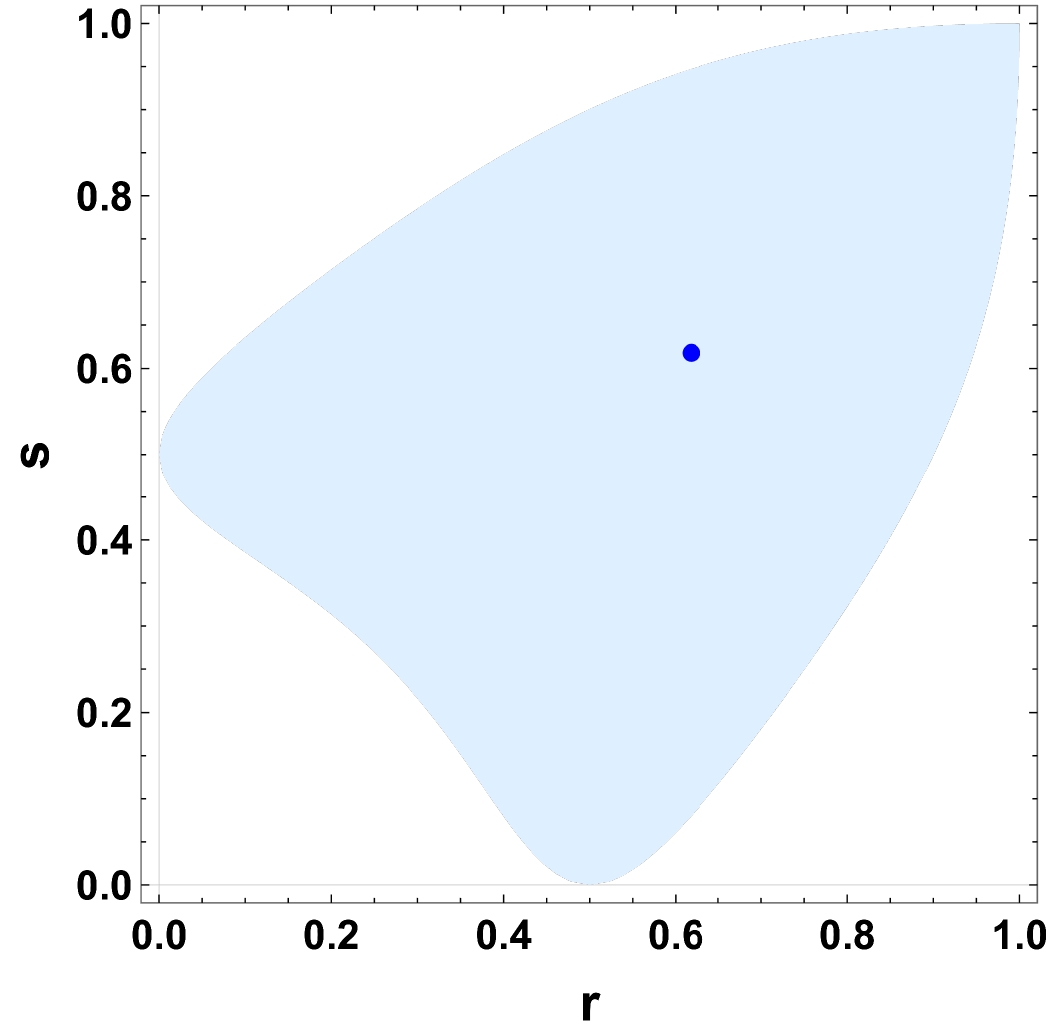}
			\caption{In shaded (light blue) region $\Omega^{\ast}(\mathbf{r},\mathbf{s})$ is concave; bold (blue) dot in the middle is the point $(\mathbf{r},\mathbf{s})=(\frac{\sqrt{5}-1}{2} ,\frac{\sqrt{5}-1}{2})$ where the function takes the maximum value.}
			\label{SupFig1b}
		\end{subfigure}
		
	\end{tabularx}
	\caption{Region of concavity for the function  $\Omega^{\ast}(\mathbf{r},\mathbf{s})$ }
	\label{SupFig1}
\end{figure}
Next, we constructed a concave cover $\mathcal{E}^{\ast}(\mathbf{r},\mathbf{s})$ of the function $\Omega^{\ast}(\mathbf{r},\mathbf{s})$. Recall that by definition a concave cover of the function $\Omega^{\ast}(\mathbf{r},\mathbf{s})$ is the lowest-valued concave function that overestimates or equals $\Omega^{\ast}(\mathbf{r},\mathbf{s})$ in its domain. We numerically computed the concave cover by applying the ConvexHullMesh function provided in MATHEMATICA. The concave cover $\mathcal{E}^{\ast}(\mathbf{r},\mathbf{s})$ is shown in Fig.(\ref{SupFig2a}) and the part of domain $\mathcal{R}^{\ast}$ where $\mathcal{E}^{\ast}(\mathbf{r},\mathbf{s})=\Omega^{\ast}(\mathbf{r},\mathbf{s})$ is shown as a shaded region in Fig.(\ref{SupFig2b}). Figure one of the main text is obtained by superposing Fig.(\ref{SupFig1b}) and Fig.(\ref{SupFig2b}). 

By a similar procedure one can construct the concave cover $\mathcal{E}(\mathbf{r},\mathbf{s})$ to any function $\Omega(\mathbf{r},\mathbf{s})$, and then the corresponding region $\mathcal{R}$ where $\mathcal{E}(\mathbf{r},\mathbf{s})=\Omega(\mathbf{r},\mathbf{s})$.
\begin{figure}[t!]
	\begin{tabularx}{\linewidth}{CC}
		\begin{subfigure}[b]{0.4\textwidth}
			\includegraphics[width=\linewidth]{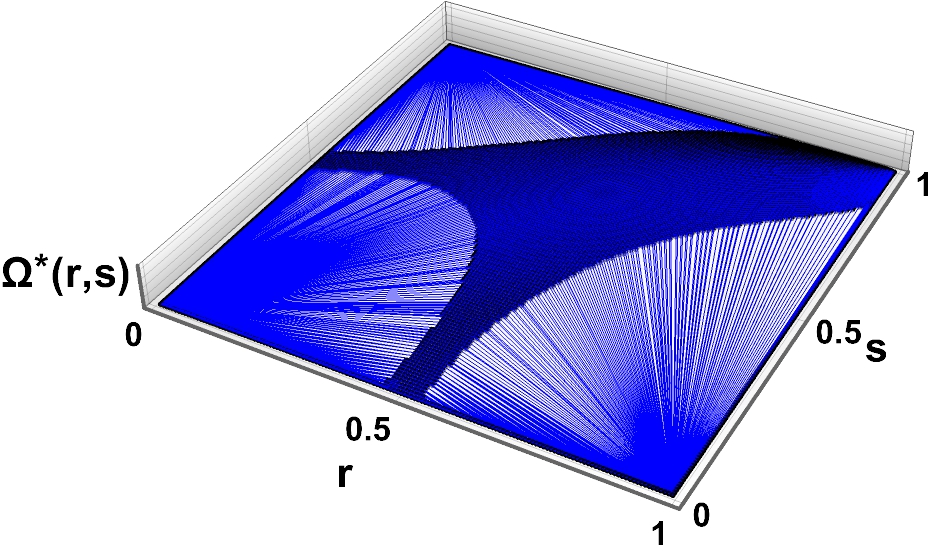}
			\caption{Concave cover $\mathcal{E}^{\ast}(\mathbf{r},\mathbf{s})$ of function $\Omega^{\ast}(\mathbf{r},\mathbf{s})$}
			\label{SupFig2a}
		\end{subfigure}
		&	
		\begin{subfigure}[b]{0.3\textwidth}
			\includegraphics[width=\linewidth]{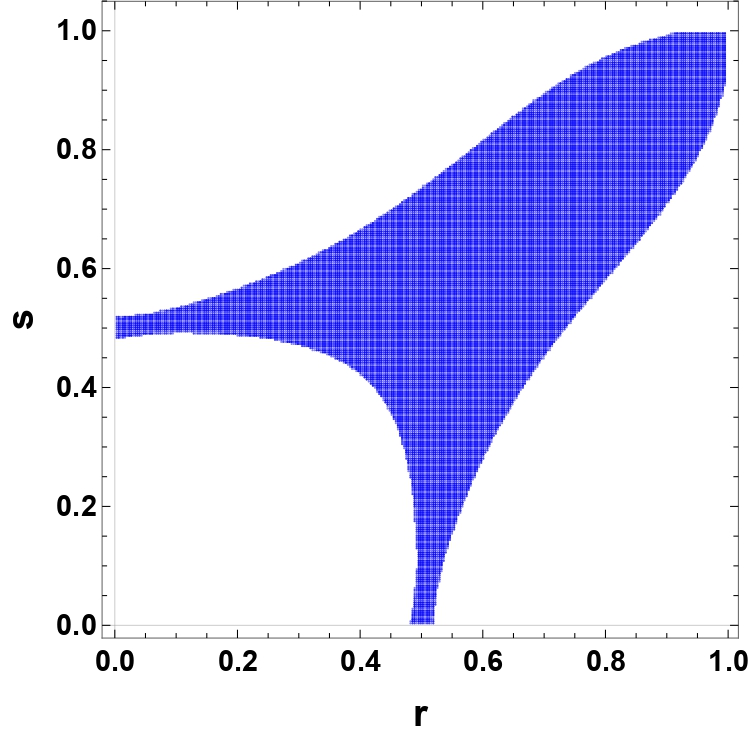}
			\caption{Shaded (Blue) region shows $\mathcal{R}^{\ast}$ where $\mathcal{E}^{\ast}(\mathbf{r},\mathbf{s})=\Omega^{\ast}(\mathbf{r},\mathbf{s})$}
			\label{SupFig2b}
		\end{subfigure}
		
	\end{tabularx}
	\caption{Construction of concave cover $\mathcal{E}^{\ast}(\mathbf{r},\mathbf{s})$ and region $\mathcal{R}^{\ast}$}
	\label{SupFig2}
\end{figure}

\section*{ Appendix~(B): Details of the proof for Lemma-$2$ in the main text}  \label{apnB}

Let us now consider the following single parameter family of functions, of the two variables $\mathbf{r}$ and $\mathbf{s}$, introduced in the construction of the proof of Lemma-2 in the main text
\begin{align}
	\Omega_{\nu}(\mathbf{r},\!\mathbf{s}) \!&=\!\Omega^{\ast}(\mathbf{r},\!\mathbf{s}) \!+\nu~\!p(\!+1\vert A_0\!) \!+(1\!-\!\nu)~\!p(\!-1\vert A_0\!)\!-\!1/2,\nonumber\\
	\!&=\!\Omega^{\ast}(\mathbf{r},\!\mathbf{s})\!+\nu~\!(\mathbf{s}\!-\!\mathbf{r}\mathbf{s}) \!+(1\!-\!\nu)~\!(1\!-\!\mathbf{s}\!+\!\mathbf{r}\mathbf{s})\!-\!1/2, \label{Eq1}
\end{align}
where the parameter $\nu\in (0,1)$ and variables $(\mathbf{r},~\mathbf{s})\in (0,1)\times(0,1)$. Firstly, we constructed finite subset of functions from the family in~Eq.(\ref{Eq1}) containing $N \!- \!1$ members: $S_{N}=\{\Omega_{\nu}(\mathbf{r},\mathbf{s}):~\nu \in \{d/N:~d\in\{1,2,...,N \!-\! 1 \} \}\}$. Let us denote the $k$-th member of $S_{N}$ by $S_{N}(k)$ where $k\in\{1,2,,...,N \!-\! 1\}$. 
\begin{figure}[h!]
	\begin{tabularx}{\linewidth}{CCCC}
		\centering
		\begin{subfigure}[b]{0.25\textwidth}
			\centering
			\caption{$\mathcal{R}_{\frac{1}{4}}$~$(\nu=1/4)$}
			\includegraphics[width=\textwidth]{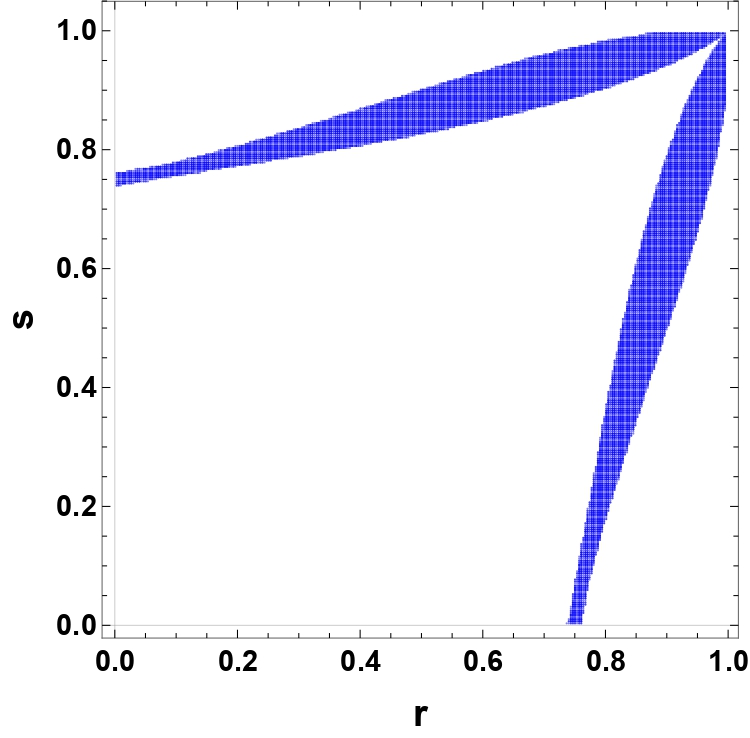}
			\label{Fig1a}
		\end{subfigure}
		%	\hfill
		&
		\begin{subfigure}[b]{0.25\textwidth}
			\centering
			\caption{$\mathcal{R}_{\frac{1}{2}}$~$(\nu=1/2)$}
			\includegraphics[width=\textwidth]{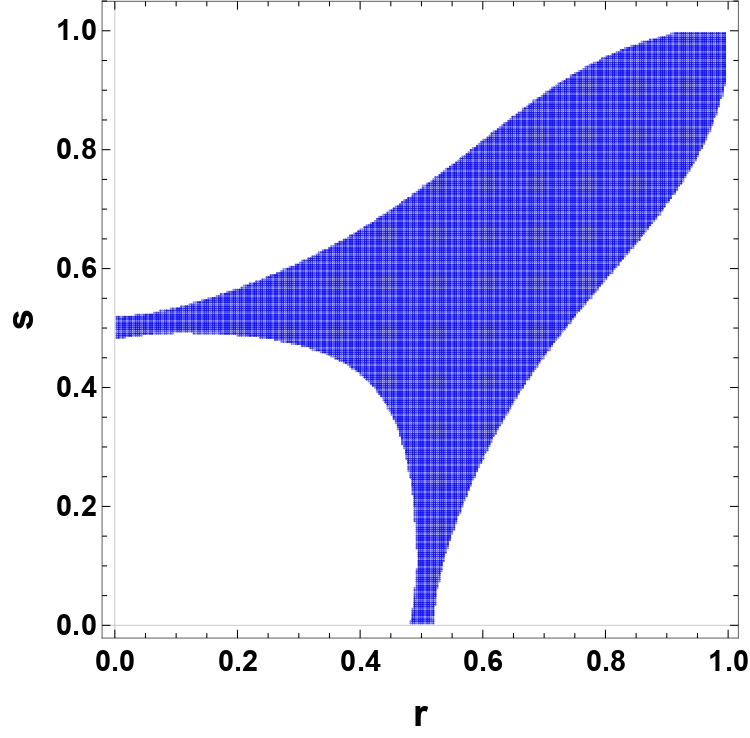}
			\label{Fig1b}
		\end{subfigure}
		%\hfill
		&
		\begin{subfigure}[b]{0.25\textwidth}
			\centering
			\caption{$\mathcal{R}_{\frac{3}{4}}$~$(\nu=3/4)$}
			\includegraphics[width=\textwidth]{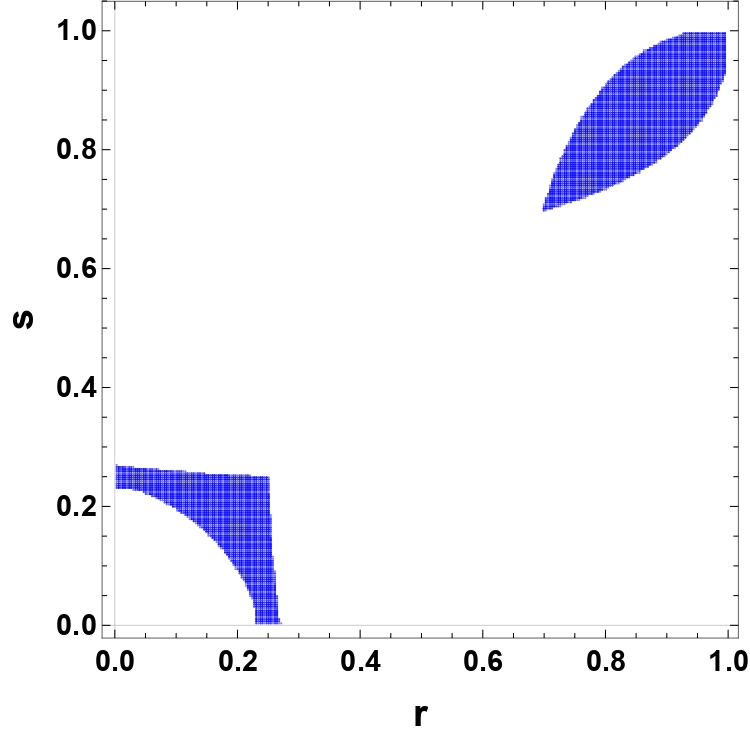}
			\label{Fig1c}
		\end{subfigure}
		%\caption{Three simple graphs}
		%	\hfill
		&
		\begin{subfigure}[b]{0.25\textwidth}
			\centering
			\caption{$\cup_{\nu} \mathcal{R}_{\nu}$}
			\includegraphics[width=\textwidth]{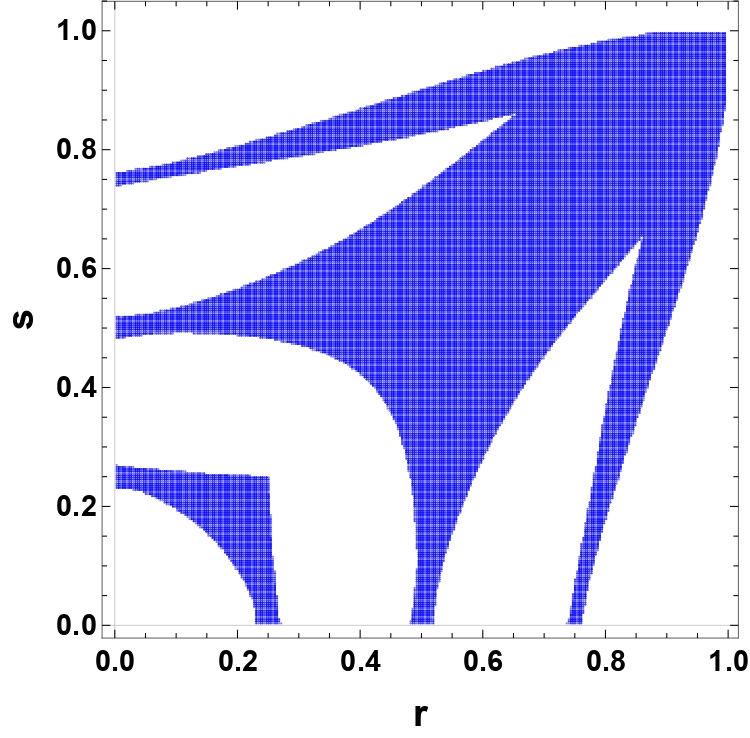}
			\label{Fig1d}
		\end{subfigure}
		
	\end{tabularx}
	\caption{Plots of regions $\mathcal{R}_{\nu}$ and $\cup_{\nu}\mathcal{R}_{\nu}$ for $\nu\in\{\frac{1}{4},\frac{1}{2},\frac{3}{4}\}$ and $N=4$.}
	\label{SupFig3}
\end{figure}

\begin{figure}[h!]
	\begin{tabularx}{\linewidth}{CCCC}
		\centering
		\begin{subfigure}[b]{0.25\textwidth}
			\centering
			\caption{$N=2$}
			\includegraphics[width=\textwidth]{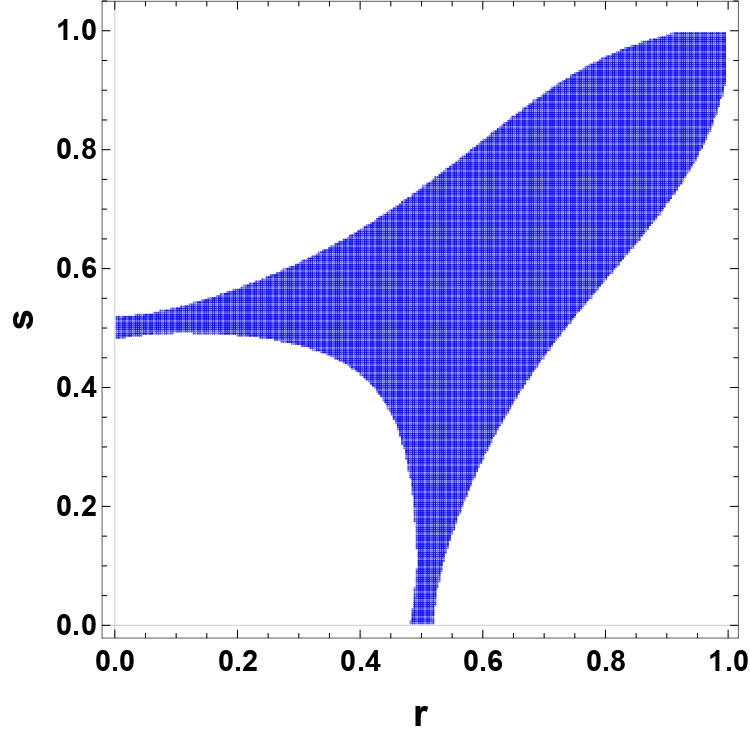}
			\label{Fig2a}
		\end{subfigure}
		&
		\begin{subfigure}[b]{0.25\textwidth}
			\centering
			\caption{$N=5$}
			\includegraphics[width=\textwidth]{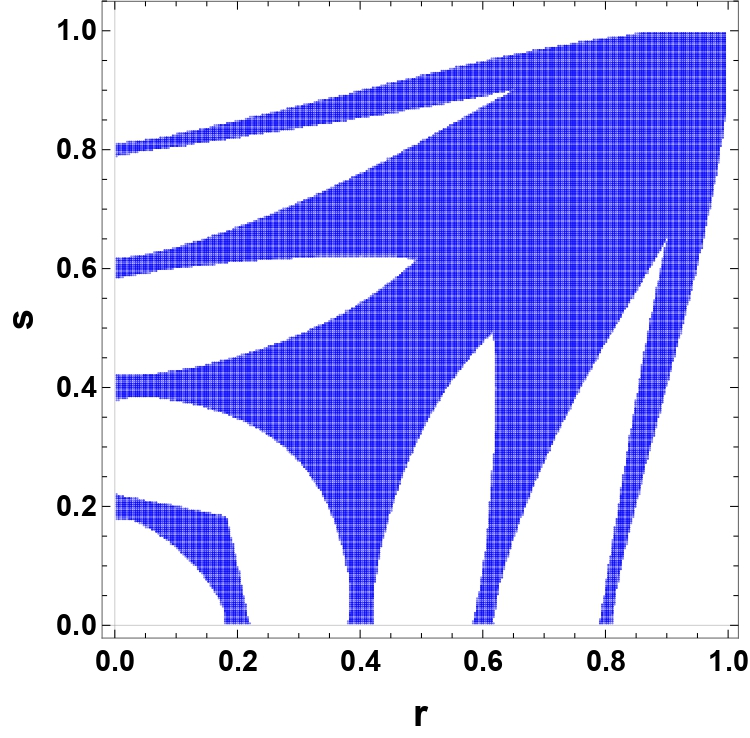}
			\label{Fig2b}
		\end{subfigure}
		%\hfill
		&
		\begin{subfigure}[b]{0.25\textwidth}
			\centering
			\caption{$N=10$}
			\includegraphics[width=\textwidth]{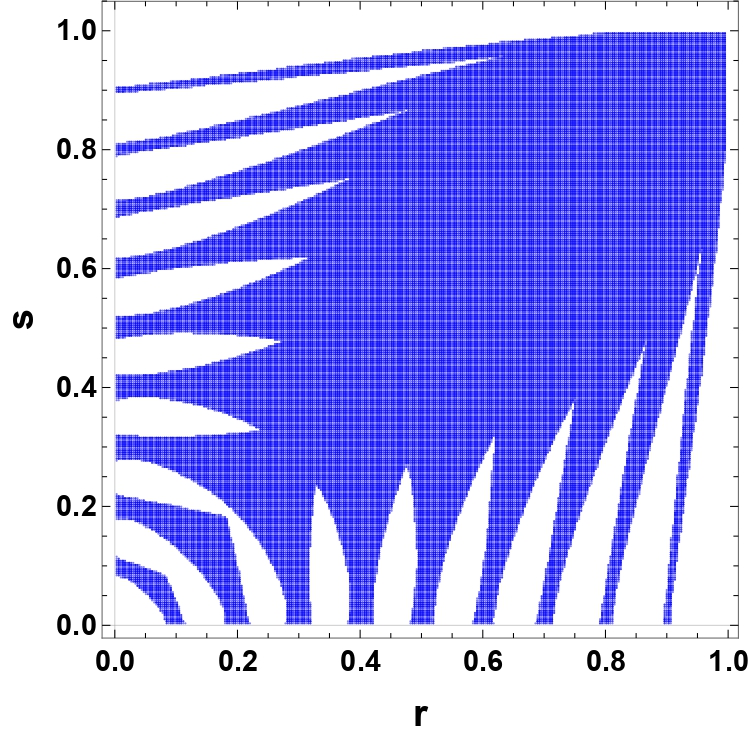}
			\label{Fig2c}
		\end{subfigure}
		%\caption{Three simple graphs}
		&
		\begin{subfigure}[b]{0.25\textwidth}
			\centering
			\caption{$N=100$}
			\includegraphics[width=\textwidth]{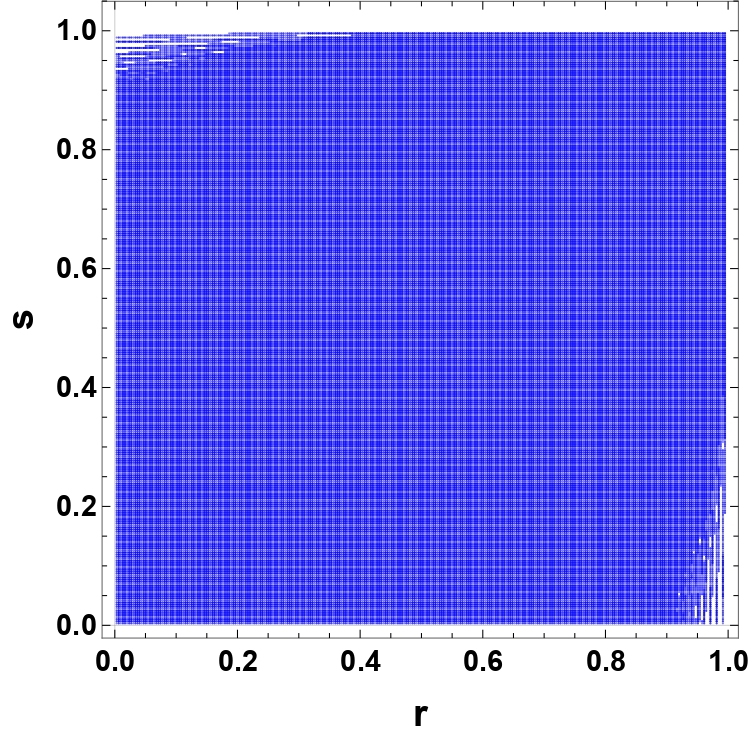}
			\label{Fig2d}
		\end{subfigure}
		
	\end{tabularx}
	\caption{Plots of region $\cup_{\nu}\mathcal{R}_{\nu}$ for different values of $N$.}
	\label{SupFig4}
\end{figure}
Secondly, we numerically constructed concave envelope functions of $S_{N}(k)$ for all values of $k$ by ConvexHullMesh function
in MATHEMATICA, then the region of domain where value of a concave envelope function coincide with the value of $S_{N}(k)$ was found as $\mathcal{R}_{\nu}$. Finally, for fixed values of $N$ the region $\cup_{\nu} \mathcal{R}_{\nu}$ was constructed. For example, the Fig.(\ref{SupFig3}) shows these region for $N=4$. By the above procedure, for a sufficiently large value of $N$ we find that $\cup_{\nu} \mathcal{R}_{\nu} = (0,1)\times (0,1)$. As an illustration see Fig.(\ref{SupFig4}) where the region $\cup_{\nu}\mathcal{R}_{\nu}$ is plotted when $N$ take values $2$, $5$, $10$, and $100$.

\end{widetext}

\end{document}